\documentclass{eptcs-modified}
\usepackage{amsmath,caption}
\usepackage{tikz,wrapfig}
\usepackage{amssymb}
\usepackage{amsthm}
\usepackage{cleveref,ulem}

\title{Comparing representations for function spaces in computable analysis}
\author{
Arno Pauly
\institute{Computer Laboratory\\ University of Cambridge, United Kingdom \\ \& \\
D\'epartement d'Informatique\\ Universit\'e Libre de Bruxelles, Belgium
\email{Arno.Pauly@cl.cam.ac.uk}}
\and
Florian Steinberg
\institute{Fachbereich Mathematik\\ TU-Darmstadt, Germany
\email{steinberg@mathematik.tu-darmstadt.de}}
}
\def\titlerunning{Comparing representations for function spaces}
\def\authorrunning{A. Pauly \& F. Steinberg}
	\newtheorem{theorem}{Theorem}
	\newtheorem{definition}[theorem]{Definition}

	\newtheorem{corollary}[theorem]{Corollary}
	\newtheorem{proposition}[theorem]{Proposition}
	\newtheorem{lemma}[theorem]{Lemma}

	\newtheorem{example}[theorem]{Example}
	\newcommand{\Wd}[1]{\left[#1\right]}
	\newcommand{\dom}{\operatorname{dom}}
	\newcommand{\adv}{\operatorname{Adv}}

	\newcommand{\id}{\textnormal{id}}
	\newcommand{\Cantor}{{\{0, 1\}^\mathbb{N}}}
	\newcommand{\Baire}{{\mathbb{N}^\mathbb{N}}}
	\newcommand{\uint}{{[0,1]}}
	
	\newcommand{\hide}[1]{}
	\newcommand{\mto}{\rightrightarrows}

	\newcommand{\lpo}{\textrm{LPO}}
	
	\newcommand{\leqW}{\leq_{\textrm{W}}}

	\newcommand{\leW}{<_{\textrm{W}}}
	\newcommand{\equivW}{\equiv_{\textrm{W}}}

	\newcommand{\CCN}{\mathrm{C}_{\NN}}
	\newcommand{\demph}{\textbf}
	\newcommand{\C}{\textrm{C}}
	\newcommand{\NN}{\mathbb{N}}
	
	\newcommand{\RR}{\mathbb{R}}
	
	\newcommand{\CC}{\mathbb{C}}
	\newcommand{\DD}{\mathbb{D}}
	
	\newcommand{\OO}{\mathcal{O}}
	\newcommand{\XX}{\mathbf{X}}
	\newcommand{\YY}{\mathbf{Y}}
	\newcommand{\EE}{\mathcal E}
	\newcommand{\SF}{\mathcal S}
	\newcommand{\TF}{\mathcal D}
	\newcommand{\A}{\mathcal A}
	\newcommand{\abs}[1]{\left|#1\right|}
	\newcommand{\germs}{\OO}
	\newcommand{\Sum}{\operatorname{Sum}}
	\newcommand{\analytic}{\mathcal C^\omega(D)}
	\newcommand{\cont}{\mathcal C(D)}
	\newcommand{\Advc}{\operatorname{Adv}_{\C^\omega}}
	\newcommand{\Advg}{\operatorname{Adv}_{\germs}}
	\newcommand{\Count}{\operatorname{Count}}
	\newcommand{\Bound}{\operatorname{Bound}}
	\newcommand{\Germ}{\operatorname{Germ}}
	\newcommand{\Diff}{\operatorname{Diff}}
	\newcommand{\supp}{\operatorname{supp}}
	\newcommand{\Monic}{\operatorname{Monic}}
	\newcommand{\Zeros}{\operatorname{Zeros}}
	
	\newcommand{\dbnd}{\operatorname{Dbnd}}
	\newcommand{\card}[1]{\#{#1}}
	\newcommand{\norm}[1]{\left\|#1\right\|}
	\newcommand{\bnorm}[1]{\big\|#1\big\|}

\begin{document} \theoremstyle{definition}

	\maketitle

	\begin{abstract}
		This paper compares different representations (in the sense of computable analysis) of a number of function spaces that are of interest in analysis.
		In particular subspace representations inherited from a larger function space are compared to more natural representations for these spaces.
		The formal framework for the comparisons is provided by Weihrauch reducibility.

		The centrepiece of the paper considers several representations of the analytic functions on the unit disk and their mutual translations.
		All translations that are not already computable are shown to be Weihrauch equivalent to closed choice on the natural numbers.
		Subsequently some similar considerations are carried out for representations of polynomials.
		In this case in addition to closed choice the Weihrauch degree $\lpo^*$ shows up as the difficulty of finding the degree or the zeros.
		As a final example, the smooth functions are contrasted with functions with bounded support and Schwartz functions.
		Here closed choice on the natural numbers and the $\lim$ degree appear.
	\end{abstract}
	\tableofcontents

	\section{Introduction}
		In order to make sense of computability questions in analysis, the spaces of objects involved have to be equipped with representations:
		A representation determines the information that is provided (or has to be provided) when computing on these objects.
		When changing from a more general to more restrictive setting, there are two options:
		Either to merely restrict the scope to the special objects and retain the representation, or to introduce a new representation containing more information.

		As a first example of this, consider the closed subsets of $\uint^2$ and the closed convex subsets of $\uint^2$ (following \cite{paulyleroux}).
		The former are represented as enumerations of open balls exhausting the complement.
		The latter are represented as the intersection of a decreasing sequence of rational polygons.
		Thus, prima facie the notions of a \emph{closed set which happens to be convex} and a \emph{convex closed set} are different.
		In this case it turns out they are computably equivalent after all (the proof, however, uses the compactness of $\uint^2$).

	\subsection{Summary of the results}
		This paper presents different examples of the same phenomenon:
		In \Cref{sec:analytic} the difference between an \emph{analytic function} and a \emph{continuous functions that happens to be analytic} is investigated for functions on a fixed compact domain.
		It is known that these actually are different notions.
		The results quantify how different they are using the framework of Weihrauch reducibility.
		The additional information provided for an analytic function over a continuous function can be expressed by a single natural number.
		Thus, this is an instance of computation with discrete advice as introduced in \cite{MR2915702}.
		Finding this number is Weihrauch equivalent to $\CCN$.
		This means that while the number can be chosen to be falsifiable (i.e.~wrong values can be detected), this is the only computationally relevant restriction on how complicated the relationship between object and associated number can be.
		The results are summarized in \Cref{figure:reductions} on Page \pageref{figure:reductions}

		\Cref{sec:polynomials} considers \emph{continuous functions that happen to be polynomials} versus \emph{analytic functions that happen to be polynomials} versus \emph{polynomials}.
		All translations turn out to be either computable, or Weihrauch equivalent to one of the two well-studied principles $\CCN$ and $\lpo^*$.
		The results are summarized in \Cref{figure:reductions for polynomials} on Page \pageref{figure:reductions for polynomials}.

		The last \Cref{sec:schwartz} changes the setting in that it swaps the compact subset of the complex plane as domain for the real line.
		It contrasts the spaces of smooth functions, Schwartz functions and bump functions.
		While going from smooth (or Schwartz) to a bump function is equivalent to $\CCN$, going from a smooth function that happens to be Schwartz to a Schwartz function is equivalent to the Weihrauch degree $\lim$.
		This degree captures the Halting problem.
		In particular it follows that there is a function $f\in\C^\infty(\RR)$ that decays faster than any polynomial (i.e. $f\in\SF$) and is computable as element of $C^\infty(\RR)$, but as element of $\SF$ is not only  non-computable, but computes the Halting problem.

		We briefly mention two alternative perspectives on the phenomenon:
		First, recall that in intuitionistic logic a negative translated statement behaves like a classical one, and that double negations generally do not cancel.
		In this setting the difference boils down to considering either analytic functions or continuous functions that are not not analytic.
		Second, from a topological perspective, Weihrauch equivalence of a translation to $\CCN$ implies that the topologies induced by the representations differ.
		Indeed, the suitable topology on the space of analytic functions is not just the subspace topology inherited from the space of continuous functions but a direct limit topology.

		An extended abstract based on this paper can be found as \cite{pauly-steinberg-csr}.

		\subsection{Represented spaces}
			This section provides a very brief introduction to the required concepts from computable analysis.
			For a more in depth introduction and further information, the reader is pointed to the standard textbook in computable analysis \cite{MR1795407}, and to \cite{pauly-synthetic}.
			Also, \cite{MR1005942} should be mentioned as an excellent source, even though the approach differs considerably from the one taken here.
		
			Recall that a \demph{represented space} $\mathbf{X} = (X, \delta_\mathbf{X})$ is given by a set $X$ and a partial surjection $\delta_\mathbf{X}: \subseteq \Baire \to X$ from Baire space onto it.
			The elements of $\delta^{-1}_{\XX}(x)$ should be understood as encodings of $x$ and are called the \demph{$\XX$-names} of $x$.
			Each represented spaces inherits a topology from Baire space:
			The final topology of the chosen representation.
			We usually refrain from mentioning the representation of a represented space in the same way as the topology of a topological space is usually not mentioned.
			For instance the set of natural numbers is regarded as a represented space with the representation $\delta_{\NN}(p) := p(0)$.
			Therefore, from now on denote by $\NN$ not only the set or the topological space, but the \demph{represented space of natural numbers}.
			If a topological space is to represented, the representation should be chosen such that it fits the topology as good as possible.
			For instance for the case $\NN$ above, the final topology of the representation is the discrete topology.

			If $\XX$ is a represented space and $Y$ is a subset of $\XX$, then $Y$ can be turned into a represented space by equipping it with the range restriction of the representation of $\XX$.
			Denote the represented space arising in this way by $\XX|_{Y}$.
			We use the same notation $\XX|_{\YY}$ if $\YY$ is a represented space already.
			In this case, however, no information about the representation of $\YY$ is carried over to $\XX|_{\YY}$.
			
			Recall that a multivalued function $f$ from $X$ to $Y$ (or $\XX$ to $\YY$) is an assignment that assigns to each element $x$ of its domain a set $f(x)$ of acceptable return values.
			Multivaluedness of a function is indicated by $f:\XX\mto\YY$.
			The domain of a multivalued function is the set of elements such that the image is not empty.
			Furthermore, recall that $f:\subseteq \XX\to\YY$ indicates that the function $f$ is allowed to be partial, i.e. that its domain may be a proper subset of $\XX$.

			\begin{definition}
				A partial function $F : \subseteq \Baire \to \Baire$ is a \demph{realizer} of a multivalued function $f : \subseteq\mathbf{X} \mto \mathbf{Y}$ if $\delta_\mathbf{Y}(F(p)) \in f(\delta_\mathbf{X}(p))$ for all $p \in \delta_{\XX}^{-1}(\dom(f))$ (compare \Cref{figure:diagram}).
			\end{definition}
			
			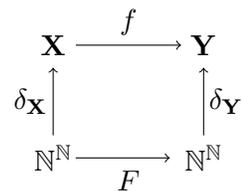
\begin{wrapfigure}{r}{3.4cm}
				\vspace{-.5cm}
				\begin{tikzpicture}
					\node at (0,0) {$\Baire$};
					\node at (0,1.5) {$\XX$};
					\node at (2,0) {$\Baire$};
					\node at (2,1.5) {$\YY$};
					\draw[->] (0.3,1.5) -- (1.7,1.5);				
					\draw[->] (0,0.3) -- (0,1.2);
					\draw[->] (2,0.3) -- (2,1.2);
					\draw[->] (0.3,0) -- (1.55,0);
					\node at (1,1.8) {$f$};
					\node at (-.3,.75) {$\delta_{\XX}$};
					\node at (2.3,.75) {$\delta_{\YY}$};
					\node at (1,-.3) {$F$};
				\end{tikzpicture}
				\caption{A diagram}\label{figure:diagram}
				\vspace{-.8cm}
			\end{wrapfigure}
			A (multivalued) function between represented spaces is called computable if it has a computable realizer, where computability on Baire space is defined via oracle Turing machines (as in e.g.~\cite{MR1673610}) or via Type-2 Turing machines (as in e.g.~\cite{MR1795407}).
			We call a (multivalued) function between represented spaces continuous, if it has a continuous realizer.
			For single valued functions on admissibly represented spaces (in the sense of Schr\"oder \cite{schroder}), this notion coincides with topological continuity.
			All representations discussed in this paper are admissible.		

			The remainder of this section introduces the represented spaces that are needed for the content of the paper.

			\subsubsection*{Sets of natural numbers.}

				Let $\OO(\NN)$ resp.\ $\A(\NN)$ denote the \demph{represented spaces of open} resp.\ \demph{closed subsets of $\NN$}.
				The underlying set of both $\OO(\NN)$ and $\A(\NN)$ is the power set of $\NN$.
				The representation of $\OO(\NN)$ is defined by
				\[ \delta_{\OO(\NN)}(p) = O \quad\Leftrightarrow\quad O = \{p(n)-1\mid n\in\NN, p(n) > 0\}. \]
				That is: A $\OO(\NN)$-name of a set of natural numbers is an enumeration of the set, however the enumeration is allowed to not return an element of the set in each step (otherwise no finite set would get a name).
				The closed sets $\A(\NN)$ are represented as complements of open sets:
				\[ \delta_{\A(\NN)}(p) = A \quad\Leftrightarrow\quad \delta_{\OO(\NN)}(p) = A^c. \]
				I.e. a $\A(\NN)$-name of a set of natural numbers is an enumeration of the complement.


			\subsubsection*{Computable metric spaces, $\RR$, $\CC$, $\cont$.}
				Given a triple $\mathcal M = (M,d,(x_n)_{n\in \NN})$ such that $(M,d)$ is a separable metric space and $x_n$ is a dense sequence, turn $\mathcal M$ into a represented space by equipping it with the representation
				\[ \delta_{\mathcal M}(p) = x \quad \Leftrightarrow\quad \forall n\in \NN: d(x,x_{p(n)}) < 2^{-n}. \]
				This is the canonical way of considering $\RR$, $\RR^d$ and $\CC$ as represented spaces; the dense sequences are standard enumerations of the rational elements.

				The space $\cont$ of continuous functions on a compact subset $D$ of $\RR^d$ is a separable metric space and thus a represented space.
				The metric is the one induced by the supremum norm and the dense sequences are standard enumerations of the polynomials with rational coefficients.
				The computable Weierstra\ss\ approximation theorem states that a function is computable as element of $C([0,1])$ if and only if it is computable in the sense of realizers as a function between the represented spaces $\uint$ and $\RR$ respectively.

			\subsubsection*{Sequences in a represented space.}

				For a represented space $\XX$ there is a canonical way to turn the set of sequences in $\XX$ into a \demph{represented space $\XX^\NN$}:
				Let $\langle\cdot,\cdot\rangle:\NN\times \NN \to \NN$ be a standard paring function (i.e. bijective, recursive with recursive projections).
				Define a function $\langle\cdot\rangle:\left(\Baire\right)^\NN\to \Baire$ by
				\[ \langle (p_k)_{k\in \NN}\rangle(\langle m,n\rangle) := p_m(n). \]
				For a represented space $\XX$ define a representation of the set $X^\NN$ of the sequences in the set $X$ underlying $\XX$ by
				\[ \delta_{\XX^\NN}(\langle (p_k)_{k\in\NN}\rangle) =(x_k)_{k\in\NN} \quad \Leftrightarrow \quad \forall m\in\NN: \delta_{\XX}(p_m) = x_m. \]
				I.e. $p$ is a name of $(x_m)_{m\in\NN}$ if for each fixed $m$ the mapping $n\mapsto p(\langle m,n\rangle)$ is a name of $x_m$.
				In particular the spaces $\RR^\NN$ and $\CC^\NN$ of real and complex sequences are considered represented spaces in this way.
				For a partial, multivalued function $f:\subseteq\XX\mto\YY$ let $f^\NN:\subseteq\XX^\NN\mto\YY^\NN$ denote the function defined by $f^\NN((x_n)_{n\in\NN}):= (f(x_n))_{n\in\NN}$.

		\subsection{Weihrauch reducibility}
			This section provides a brief introduction to Weihrauch reducibility.
			The research programme of Weihrauch reducibility was formulated in \cite{MR2099383}, a more up-to-date introduction can be found in \cite{MR3350999}.

			Every multivalued function $f:\subseteq\XX\mto\YY$ corresponds to a computational task.
			Namely: \lq given information about $x$ and the additional assumption $x\in \dom(f)$ find suitable information about some $y\in f(x)$\rq.
			What information about $x$ resp.\ $f(x)$ is provided resp.\ asked for is reflected in the choice of the representations for $\XX$ and $\YY$.
			The following example of this is very relevant for the content of this paper:
			\begin{definition}
				Let \demph{closed choice on the integers} be the multivalued function $\C_\mathbb{N} : \subseteq \mathcal{A}(\mathbb{N}) \mto \mathbb{N}$ defined on nonempty sets by
				\[ y\in\C_\mathbb{N}(A)\Leftrightarrow y\in A. \]
			\end{definition}
			The corresponding task is \lq given an enumeration of the complement of a set of natural numbers and provided that it is not empty, return an element of the set\rq.
			$\C_\NN$ does not permit a computable realizer:
			Whenever a machine decides that the name of the element of the set should begin with $n$, it has only read a finite beginning segment of the enumeration.
			The next value might as well be $n$.

			From the point of view of multivalued functions as computational tasks, it makes sense to compare their difficulty by comparing the corresponding multivalued functions.
			This paper uses Weihrauch reductions as formalization of such a comparison.
			Weihrauch reductions define a rather fine pre-order on multivalued functions between represented spaces.

			\begin{wrapfigure}{r}{3.5cm}
				\vspace{-.2cm}
				\begin{tikzpicture}
					\node at (.5,4.75) {\small{name of some $y\in f(x)$}};
					\draw (-1.25,-.5) rectangle (2.25,4.25);
					\draw (0,0) rectangle (2,1);
					\node at (1,.5) {$H$};
					\draw[->] (0,-.5) -- (.25,0);
					\node at (.5,-1.25) {name of $x\in \dom(f)$};
					\draw[->] (.25,1) -- (.25,1.5);
					\node at (1.25,1.25) {name of $z$};
					\draw (0,1.5) rectangle (2,2.5);
					\node at (1,2) {$G$};
					\node at (1.25,2.75) {\small{name of $g(z)$}};
					\draw[->] (.25,2.5) -- (.25,3);
					\draw[->] (0,-1) -- (0,-.5) -- (-.5,.5) -- (-.5,2) --(-.25,2.5) --(-.25,3);
					\draw (-1,3) rectangle (1,4);
					\node at (0,3.5) {$K$};
					\draw[->] (0,4) -- (0,4.5);
					\node at (1.6,3.4) {F};
				\end{tikzpicture}
				\caption{Weihrauch reductions}\label{fig:Wheirauch reduction}
				\vspace{-.8cm}
			\end{wrapfigure}
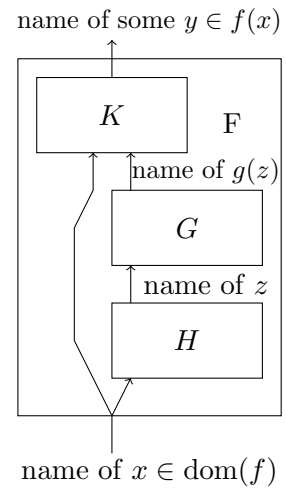
			
			\begin{definition}\label{def:weihrauch}
				Let $f$ and $g$ be partial, multivalued functions between represented spaces.
				Say that $f$ is \demph{Weihrauch reducible} to $g$, in symbols $f\leqW g$, if there are computable
				functions $K:\subseteq \Baire\times \Baire \to\Baire$ and $H:\subseteq\Baire\to\Baire$ such that whenever $G$ is a realizer of $g$, the function $F := \left (p \mapsto K(p,G(H(p)))\right )$ is a realizer for $f$.
			\end{definition}

			$H$ is called the \demph{pre-processor} and $K$ the \demph{post-processor} of the Weihrauch reduction.
			This definition and the nomenclature is illustrated in \Cref{fig:Wheirauch reduction}.
			The relation $\leqW$ is reflexive and transitive.
			We use $\equivW$ to denote that reductions in both directions exist and $\leW$ if this is not the case.
			The equivalence class of a multivalued function with respect to the equivalence relation $\equivW$ is called the \demph{Weihrauch degree} of a function $f$ and we denote it by $\left[f\right]$.
			We still use $\leqW$ for the induced partial order on the Weihrauch degrees and by abuse of notations sometimes use it to compare multivalued functions and Weihrauch degrees.
			A Weihrauch degree is called non-computable if it contains no computable function.

			The Weihrauch degree corresponding to $\C_{\NN}$ has received significant attention, e.g.~in \cite{MR2760117,MR2915694,paulymaster,mylatz,mylatzb,MR3350999,paulyoracletypetwo,mummert,paulyneumann}.
			In particular, as shown in \cite{paulydebrecht}, a function between computable Polish spaces is Weihrauch reducible to $\C_\mathbb{N}$ if and only if it is piecewise computable or equivalently is effectively $\Delta^0_2$-measurable.
			
			For the purposes of this paper, the following representatives of this degree are also relevant:
			\begin{lemma}[\cite{pauly-fouche2}]\label{lemma:cn}
				The following are Weihrauch equivalent:
				\begin{itemize}
					\item $\CCN$, that is closed choice on the natural numbers.
					\item $\max : \subseteq \mathcal{O}(\mathbb{N}) \to \mathbb{N}$ defined on the bounded sets in the obvious way.
					\item $\operatorname{Bound} : \subseteq \mathcal{O}(\mathbb{N}) \mto \mathbb{N}$, where $n \in \operatorname{Bound}(U)$ iff $\forall m \in U: n \geq m$.
				\end{itemize}
			\end{lemma}

			Given $p\in\Baire$ denote the support of $p$ by $\supp(p) := \{n\in\NN\mid p(n)>0\}$.
			Furthermore, for a finite set $A$ denote the number of elements of that set by $\card A$.
			\begin{lemma}\label{lemma:count}
				The function $\operatorname{Count}:\subseteq \Baire \to \mathbb{N}$, defined via
				\[ \dom(\operatorname{Count}) := \{p \in \Baire \mid \supp(p) \text{ is finite}\} \quad\text{and}\quad \operatorname{Count}(p) := \card{\supp(p)} \]
				is Weihrauch equivalent to closed choice on the naturals $\CCN$.
			\end{lemma}

			\begin{proof}
				First construct the Weihrauch reduction that proves $\CCN\leqW\operatorname{Count}$:
				Let the pre-processor $H$ be the function sending some $p\in \dom(\CCN)$ to the function that returns $1$ on input $n$ whenever its support up to $n$ has fewer elements than the least element that has not been excluded from the set by the first $n$ elements of the enumeration of the complement.
				This function is computable as can be seen from its recursive definition:
				\[ H(p)(n) := \begin{cases} 1 & \text{ if }\min\{m\mid\forall j\leq n: p(j)-1\neq m\} > \card{\{j < n \mid H(p)(j) = 1\}} \\ 0 & \text{ otherwise.}\end{cases} \]
				$H(p)$ has finite support, since the set described by $p$ is nonempty:
				There is some $m$ never shows up as value of $p$ and by definition the support of $H(p)$ does not outgrow that number.
				Applying a realizer of $\operatorname{Count}$ to $H(p)$ returns an encoding of the least element of the set encoded by $p$.
				To obtain a Weihrauch reduction just pass this number on to be the output via $K(p,q)(n) := q(0)$.
				
				For the opposite direction, i.e. $\Count\leqW\CCN$ use \Cref{lemma:cn} and replace $\CCN$ by $\max$.
				Define the pre-processor $H$ of a Weihrauch reduction $\Count\leqW\max$ as follows:
				\[ H(p)(n) := \begin{cases} n+1 & \text{ if }p(n) > 0 \\ 0 & \text{otherwise.} \end{cases} \]
				This means that $H(p)$ is a $\OO(\NN)$ name of $\supp(p)$.
				Applying $\max$ will give the maximal element of the support.

				Define the post-processor $K$ to be the function
				\[ K(p,q)(n) := \card{\{m\mid m \leq q(0)\text{ and } p(m)>0\}}. \]
				This function is computable and since $q(0)$ will always be the maximal element of the support of $p$, the composition counts the support of $p$ and is a Weihrauch reduction.
			\end{proof}

			In Section \ref{sec:polynomials} another non-computable Weihrauch degree is encountered: $\lpo^*$.
			Here, $\lpo$ is short for \lq limited principle of omniscience\rq.
			We refrain from stating $\lpo^*$ explicitly as it would need more machinery than we introduced.
			Instead we characterize it by specifying the representative that is used in the proofs:
			\begin{proposition}[{\cite[Korollar 3.1.4.6]{paulymaster}}]\label{resu:lpo*}
				The function $\min : \Baire \to \mathbb{N}$ defined on all of Baire space in the obvious way is a representative of the Weihrauch degree $\lpo^*$.
			\end{proposition}

			A third and final Weihrauch degree making an appearance as the degree of a translation is the degree $\lim$ encountered in Section \ref{sec:schwartz}.

			\begin{definition}
				Let $\mathbf{X}$ be a computable metric space.
				Then $\lim_\mathbf{X} : \subseteq \mathbf{X}^\mathbb{N} \to \mathbf{X}$ maps a converging sequence to its limit.
			\end{definition}

			As shown in \cite{MR2915694}, $\lim_\mathbb{N} \equivW \C_\mathbb{N}$.
			In general, it holds that $\lim_\mathbf{X} \leqW \lim_\Baire \equivW \lim_\Cantor$, and whenever $\mathbf{Y}$ is a subspace of $\mathbf{X}$, then $\lim_\mathbf{Y} \leqW \lim_\mathbf{X}$.
			As $\Cantor$ embeds into any computable metric space considered in this paper apart from $\mathbb{N}$, it suffices to consider the Weihrauch degree $\lim:=\big[{\lim_\Cantor}\big]$ corresponding to $\lim_\Cantor$.
			The degree $\lim$ is also complete for the effectively $\Sigma^0_2$-measurable functions \cite{MR2099383,pauly-descriptive}.

			To give more intuition as to why the Weihrauch degrees $\lpo^*$, $\C_\mathbb{N}$ and $\lim$ show up in this paper, note the following:
			All three are derived from the maybe simplest non-computable Weihrauch degree $\lpo$ via canonical closure operators defined on the Weihrauch degrees.
			$\lpo$ the Weihrauch degree of the characteristic function of the zero function in Baire space, namely $\lpo = \Wd{\chi_0}$, where
			\[ \chi_0(p) :=\begin{cases} 1 & \text{ if $p$ is the zero function, i.e. }\forall n:p(n)=0. \\ 0 & \text{ otherwise.} \end{cases} \]
			In computable analysis $\lpo$ shows up as the Weihrauch degree of the equality test for real (or complex) numbers.

			Now, $\lpo^*$ corresponds to carrying out a fixed finite but arbitrary high number of equality tests on the real or complex numbers via the operator $^*$ from \cite{paulyreducibilitylattice}.
			The operator $^\diamond$ introduced in \cite{paulyneumann} captures using the given degree an arbitrary finite number of times (without the requirement that the number is fixed in advance), and it holds that $\Wd{\CCN} = \lpo^\diamond$.
			Define one last operator on the Weihrauch degrees by setting $\widehat{\Wd f} := \Wd{f^\NN}$.
			This operation was also investigated in \cite{brattka2} and corresponds to a countable number of invocations of $f$ all performed in parallel.
			It holds that $\widehat{\lpo} = \widehat{\Wd{\CCN}} = \lim$.

			It is known that $\id_\Baire \leW \lpo\leW\lpo^*\leW\CCN \leW \lim$.

	\section{Analytic functions}\label{sec:analytic}
		\subsection{Representations of analytic functions}\label{sec:analytic functions}

			Recall that a function is analytic if it is locally given by a power series:
			\begin{definition}
				Let $D\subseteq \CC$ be a set.
				A function $f:D \to \CC$ is called \demph{analytic}, if for every $x_0\in D$ there is a neighborhood $U$ of $x_0$ and a sequence $(a_k)_{k\in \NN}\in \CC^\NN$ such that for each $x\in U\cap D$
				\[ f(x) = \sum_{k\in \NN} a_k (x-x_0)^k. \]
			\end{definition}
			The set of analytic functions is denoted by $\analytic$.
			Each analytic function is continuous, i.e. $\analytic\subseteq \cont$.
			If $D$ is open, the analytic functions on $D$ are smooth, i.e. infinitely often differentiable.
			Any analytic function can be analytically extended to an open superset of its domain.

			Recall that a \demph{germ} of an analytic function is a point of its domain together with the series expansion around said point.
			As long as the domain is connected, an analytic function is uniquely determined by each of its germs.
			The one to one correspondence of germs and analytic functions only partially carries over to the computability realm:
			It is well known that an analytic function on the unit disk is computable if and only if the germ around any computable point of the domain is computable (see for instance \cite{MR1137517}).
			However, the proofs of these statements are inherently non-uniform.
			The operations of obtaining a germ from a function and a function from a germ are discontinuous and therefore not computable (see \cite{Muller1995}).

			There is a more suitable representation for the analytic functions than the restriction of the representation of continuous functions.
			This representation has been investigated by different authors for instance in \cite{MR3377508},\cite{MR2207129},\cite{Muller1995}.
			For simplicity restrict to the case of analytic functions on the unit disk.
			From now on let $D$ denote the closed unit disk.
			And let $U_m$ denote the open ball  $B_{r_m}(0)$ of radius $r_m :=2^{\frac 1{m+1}}$ around zero.
			Note that $U_{m+1}\subseteq U_m$ and that the intersection of all $U_m$ is the unit disk.
			Recall from the introduction that the space $\cont$ of continuous functions is represented as a metric space (where $\CC$ is identified with $\RR^2$).
			\begin{definition}\label{def:analytic rep}
				Let $\analytic$ denote the \demph{represented space of analytic functions on $D$}, where the representation is defined as follows:
				A function $q\in \Baire$ is a name of an analytic function $f$ on $D$, if and only if $f$ extends analytically to the closure of $U_{q(0)}$, the extension is bounded by $q(0)$ and $n\mapsto q(n+1)$ is a name of $f\in \cont$.
			\end{definition}

			Note that the representation of $\analytic$ arises from the restriction of the representation of continuous functions by adding discrete additional information.
			This information is quantified by the advice function $\Advc:\subseteq\cont \to \NN$ whose domain are the analytic functions and that on those is defined by
			\begin{equation}\label{eq:AC}
				\tag{AC}
				\begin{split}
					\Advc(f) & := \{q(0)\mid q\text{ is a $\analytic$-name of $f$})\} \\
					& = \{m\in\NN\mid f\text{ has an analytic cont. to }U_m\text{ bounded by }m\}.
				\end{split}
			\end{equation}
			This function turns up in the results of this paper.
			In the terminology of \cite{MR3377508}, one would say that $\analytic$ arises from the restriction $\cont|_{\analytic}$ by enriching with the discrete advice $\Advc$.

			The topology induced by the representation of $\analytic$ is well known and used in analysis:
			It can be constructed as a direct limit topology and makes $\analytic$ a so called Silva-Space.
			For more information on this topology and its relation to computability and complexity theory also compare \cite{MR2207129}.

			The set $\OO$ of germs around zero, i.e. of power series with radius of convergence strictly larger than $1$ can be made a represented space in a very similar way:
			\begin{definition}
				Let $\germs$ denote the represented space of germs around zero with radii of convergence greater than $1$, where the representation is defined as follows: A function $q\in\Baire$ is a name of a power series $(a_k)_{k\in\NN}$, if and only if
				\[ \forall k\in\NN:|a_k|\leq 2^{-\frac k{q(0)+1}}q(0) \]
				and $n\mapsto q(n+1)$ is a name of the sequence $(a_k)_{k\in\NN}$ as element of $\CC^\NN$.
			\end{definition}

			As above, this representation is related to the restriction of the representation of $\CC^\NN$ by means of an advice function $\Advg:\subseteq \CC^\NN\mto\NN$ whose domain are the sequences with radius of convergence strictly larger than one and that is defined on those by
			\begin{equation}\label{eq:AG}
				\tag{AG}
				\begin{split}
					\Advg((a_k)_{k\in\NN}) & := \{q(0) \mid q\text{ is a $\germs$-name of $(a_k)_{k\in\NN}$}\}\\
					&= \{n\in \NN\mid \forall k\in \NN: \abs{a_k} \leq 2^{-\frac k{n+1}}\cdot n\}
				\end{split}
			\end{equation}
			Again, the topology induced by this representation is well known and used in analysis: It is the standard choice of a topology on the set of germs and can be introduced as a direct limit topology.

			Proofs that the following holds can be found in \cite{MR3377508} or \cite{Muller1995}:
			\begin{theorem}[computability of summation]\label{thm:analytic rep}
				The assignment
				\[ \OO\to \analytic, \quad (a_k)_{k\in\NN} \mapsto \left(x\mapsto \sum_k a_k x^k\right) \]
				is computable.
			\end{theorem}
			A proof of the following can be found in \cite{MR3377508}:
			\begin{theorem}\label{thm:differntiation}
				Differentiation is computable as mapping from $\analytic$ to $\analytic$.
			\end{theorem}

		\subsection{Summing power series}\label{sec:sub:summing power series}
			Summing a power series is not computable on $\CC^\NN$.
			Recall that $\Advg$ was the advice function of the representation $\germs$ of germs around zero of analytic functions on the unit disk.
			The computational task corresponding to this multivalued function is to find from a sequence that is guaranteed to have radius of convergence bigger than one a constant witnessing the exponential decay of the absolute value of the coefficients (compare \cref{eq:AG} on page \pageref{eq:AG}).
			\Cref{thm:analytic rep} states that summation is computable on $\germs$.
			Therefore, the advice function $\Advg$ can not be computable.
			The following theorem classifies the difficulty of summing power series and $\Advg$ in the sense of Weihrauch reductions.

			\begin{theorem}\label{thm:main germs}
				The following are Weihrauch-equivalent:
				\begin{itemize}
					\item \textbf{$\CCN$}, that is: Closed choice on the naturals.
					\item \textbf{$\operatorname{Sum}$}, that is: The partial mapping from $\CC^\NN$ to $\cont$ defined on the sequences with radius of convergence strictly larger than one by
					\[ \Sum((a_k)_{k\in\NN})(x) := \sum_{k\in\NN} a_k x^k. \]
					I.e. summing a power series.
					\item \textbf{$\Advg$}, that is: The function from \cref{eq:AG} on page \pageref{eq:AG}. I.e. obtaining the constant from the series.
				\end{itemize}
			\end{theorem}

			\begin{proof}
			Build a Weihrauch reduction circle:
				\begin{description}
					\item[$\CCN\leqW\Sum$:]
					\Cref{lemma:count} permits to replace $\CCN$ by $\Count$.

					The Weihrauch reduction $\Count\leqW\Sum$ can be constructed as follows:
					Let the pre-processor be a realizer of the computable mapping that assigns to $p \in \Baire$ the sequence $(a_k)_{k\in\NN}\in \CC^\NN$ defined by
					\[ a_k :=\begin{cases} 1 & \text{ if }p(k)>0 \\ 0 &\text{ if }p(k)=0\end{cases}. \]
					Note that $p\in \dom(\Count)$ means that $p$ has a finite support, and the radius of convergence of $(a_k)_{k\in\NN}$ is strictly bigger than one (it is infinite).
					Applying a realizer of $\Sum$ will result in a name of the corresponding function $f$.
					From the definition of $(a_k)_{k\in\NN}$ it is clear that $f(1) = \Count(p)$.
					Therefore, the post-processor can be chosen as the second projection composed with a realizer of the evaluation in $1$, which is well known to be computable on the continuous functions.
				\item[$\Sum\leqW\Advg$:]
					Let the pre-processor be the identity.
					Note that an element of $\Advg((a_k)_{k\in\NN})$ and a $\CC^\NN$-name of $(a_k)_{k\in\NN}$ can easily be put together to an $\germs$-name of $(a_k)_{k\in\NN}$.
					Thus the post-processor can be chosen to be the composition of this mapping and a realizer of the summation mapping on $\OO$, which is computable by \Cref{thm:analytic rep}.
				\item[$\Advg\leqW\CCN$:]
					Let the pre-processor be the function that maps a given name $p$ of $(a_k)_{k\in\NN}\in\CC^{\omega}$ to an $\A(\NN)$-name of the set $\Advg((a_k)_{k\in\NN})$.
					Note that an enumeration of the complement of this set can be extracted from $p$ as follows:
					For all $k$ and $m\in\NN$ dovetail the test $\abs{a_k}>2^{-\frac k{m+1}}m$.
					If it holds for some $k$, return $m$ as an element of the complement.
					This procedure indeed produces a list of the complement of $\Advg((a_k)_{k\in\NN})$:
					If the above does not hold for any $k$, then $\abs{a_k}\leq 2^{-\frac k{m+1}}m$ for all $k$ and $m$ is an element of $\Advg((a_k)_{k\in\NN})$.

					Applying closed choice to this set will give result in a valid return value.
					Thus, choose the post-processor to be the second projection.
				\end{description}
			\end{proof}

		\subsection{Differentiating analytic functions}

			In \Cref{sec:analytic functions} we remarked that it is not possible to compute the germ of an analytic function just from a name as continuous function.
			The proof in \cite{Muller1995} that this is in general impossible, however, argues about analytic functions on an interval.
			The first lemma of this chapter proves that for analytic functions on the unit disk it is possible to compute a germ if its base point is well inside of the domain.
			We only consider the case where the base point is zero, but the proof works whenever a lower bound on the distance of the base point to the boundary of the disk is known.

			\begin{lemma}\label{lemma:computing germ}
				$\operatorname{Germ}$, that is: The partial mapping from $\cont$ to $\CC^\NN$ defined on analytic functions by mapping them to their series expansion around zero, is computable.
			\end{lemma}

			\begin{proof}
				Remember that the Cauchy integral formula states that for an analytic function $f$ on the closed unit disk:
				\[ f^{(k)}(0) = \frac{k!}{2\pi i} \int_0^{2\pi} f\left(e^{i t}\right)e^{-i t(k+1)} dt \]
				It is well known that the integral is computable from a name of the function $f\in \cont$.
				This works uniformly in $k$.
			\end{proof}

			The next theorem is very similar to \Cref{thm:main germs}.
			Both the advice function $\Advc$ and computing a germ around a boundary point are shown to be Weihrauch equivalent to $\CCN$.
			Note that the coefficients of the series expansion $(a_k)_{k\in\NN}$ of an analytic function $f$ around a point $x_0$ are related to the derivatives $f^{(k)}$ of the function via $k!a_k = f^{(k)}(x_0)$.
			Therefore, computing a series expansion around a point is equivalent to computing all the derivatives in that point.

			\begin{theorem}\label{thm:main functions}
				The following are Weihrauch equivalent:
				\begin{itemize}
					\item \textbf{$\CCN$}, that is closed choice on the naturals.
					\item \textbf{$\Diff_1$}, that is the partial mapping from $\cont$ to $\CC$ defined on analytic functions by
					\[ \Diff_1(f) := f'(1). \]
					I.e. evaluating the derivative of an analytic function in $1$.
					\item \textbf{$\Advc$}, that is the function from \cref{eq:AC}. I.e. obtaining the constant from the function.
				\end{itemize}
			\end{theorem}

			\begin{proof}
				By building a circle of Weihrauch reductions:
				\begin{description}
					\item[$\CCN\leqW\Diff_1$:]
						Use \Cref{lemma:count} and show $\Count\leqW\Diff_1$ instead.
						Let $p$ be an element of the do- \quad

						\noindent\begin{minipage}{.6\textwidth}
							main of $\Count$.
							That is: $\supp(p)=\{n\in\NN\mid p(n)>0\}$ is finite.
							Define a sequence of analytic functions $f_n:D\to \CC$ by
							\[ f_n(x) :=\left(x-x_n\right)^{-2^{n+1}},\]
							where
							\[ x_n:=1+2^{\frac{n+1}{2^{n+1}+1}} \]
							(see \Cref{fig:the functions fn}). The sequence is carefully chosen such that
						\end{minipage}\hspace{.25cm}
						\begin{minipage}{.3\textwidth}
							\vspace{-.1cm}
							\begin{tikzpicture}
								\begin{scope}
									\clip (0,0) rectangle (5,2);
									\draw[color=black,thick] plot[samples=100,domain=0:1.5] ({2.5*\x},{(\x-1-2^(1/3))^(-2)*2.5});
									\draw[color=blue,] plot[samples=100,domain=0:1.5] ({2.5*\x},{(\x-1-2^(2/5))^(-4)*2.5});
									\draw[color=green] plot[samples=100,domain=0:1.5] ({2.5*\x},{(\x-1-2^(3/9))^(-8)*2.5});
									\draw[color=yellow] plot[samples=100,domain=0:1.2] ({2.5*\x},{(\x-1-2^(4/17))^(-16)*2.5});
								\end{scope}
								\draw[->] (-.2,0) -- (3.5,0);
								\draw[->] (0,-.2) -- (0,2) node[above] {};
								\node at (0,-.4) {$0$};
								\draw (2.5,.1) -- (2.5,-.1) node[below] {$1$};
								\node at (3.75,1.75) {\textcolor{black}{$\pmb{f_0}$}};
								\node at (3.75,1.25) {\textcolor{blue}{$f_1$}};
								\node at (3.75,.75) {\textcolor{green}{$f_2$}};
								\node at (3.75,.25) {\textcolor{yellow}{$f_3$}};
							\end{tikzpicture}
							\captionof{figure}{The functions $f_n$.}\label{fig:the functions fn}
						\end{minipage}
						\[ f_n'(1) = 1\quad\text{and}\quad \forall x\in D: \abs{f_n(x)}<2^{-n}. \]
						Furthermore, it is computable as a sequence of functions from $\cont$.

						Consider the function
						\[ f(x) := \sum_{n\in\supp(p)} f_n(x). \]
						Note that this function can be computed from $p$: To approximate $f$ by a polynomial it suffices to approximate those $f_n$ whose index is small.
						Let the pre-processor be a realizer of this assignment.

						Note that applying $\Diff_1$ to the function $f$ results in
						\[ \Diff_1(f) = f'(1) = \sum_{n\in\supp(p)} f'_n(1) =\card{\supp(p)}. \]
						Therefore, the post-processor $K(p,q):=q$ results in a Weihrauch reduction.
					\item[$\Diff_1\leqW\Advc$:]
						Let the pre-processor be the identity.

						An appropriate post-processor can be described as follows: Combine the return value of $\Advc$ on the function $f$ and the $\cont$-name of $f$ to an $\analytic$-name of $f$.
						Afterwards apply a realizer of the differentiation operator on $\analytic$ which can be chosen computable by \Cref{thm:differntiation}.
						Finally, obtain a $\cont$-name of $f'$ from the $\analytic$-name and evaluate it in $1$.
					\item[$\Advc\leqW\CCN$:] Combine computable $\operatorname{Germ} : \mathcal{C}(D) \to \mathbb{C}^\mathbb{N}$, the classification $\Advg \equivW \CCN$ from Theorem \ref{thm:main germs} and computable $(a_k)_{k\in\NN} \mapsto \left(x\mapsto \sum_k a_k x^k \right ) :  \OO\to \analytic$ from Theorem \ref{thm:analytic rep}.
				\end{description}
			\end{proof}

			Recall from the introduction that $\cont|_{\analytic}$ resp.\ $\CC^\NN|_{\germs}$ denote the represented spaces obtained by restricting the representation of $\cont$ resp.\ $\CC^\NN$ to $\analytic$, resp.\ $\germs$.
			\Cref{thm:analytic rep,thm:main germs,thm:main functions} and \Cref{lemma:computing germ} are illustrated in \cref{figure:reductions}.
			\begin{figure}
				\centering
				\begin{tikzpicture}
					\node at (-3,-1) {$\left.\cont\right|_{\mathcal{C}^\omega}$};
					\draw[->,color = blue] (-3.25,-.75) -- (-3.25,.5) node[left] {\textcolor{blue}{$\Germ$}} -- (-3.25,1.75);
					\draw[->,color = red,dashed] (-2.5,-.8) -- (-.3,-.1);
					\node at (-.65,-.6) {\textcolor{red}{$\Advc$}};
					\draw[->,color = red,dashed] (-2.5,-1.2) -- (-.3,-2.4);
					\node at (-1,-1.6) {\textcolor{red}{$\Diff_1$}};
					\draw[->,color = red,dashed] (-2.25,-1.1) -- (0,-1.1) node[below] {\textcolor{purple}{$\id$}} -- (2.25,-1.1);
					\draw[->,color = blue] (2.25,-.9) -- (-2.25,-.9);
					\node at (3,-1) {$\analytic$};
					\draw[->,color = blue] (2.5,-.8) -- (.3,-.1);
					\node at (.95,-.6) {\textcolor{blue}{$\Advc$}};
					\draw[->,color = blue] (3.25,-.75) -- (3.25,0.5) node[right] {\textcolor{blue}{$\Germ$}} -- (3.25,1.75);
					\draw[->,color = blue] (2.5,-1.2) -- (.3,-2.4);
					\node at (1,-1.6) {\textcolor{blue}{$\Diff_1$}};
					\node at (0,0) {$\NN$};
					\node at (0,-2.5) {$\CC$};
					\node at (-3,2) {$\left.\CC^\NN\right|_{\germs}$};
					\draw[->, color=red,dashed] (-2.75,1.75) -- (-2.75,0.5) node[right] {\textcolor{red}{$\Sum$}} -- (-2.75,-.75);
					\draw[->,color = blue] (2.5,1.8) -- (.3,1.1);
					\node at (.95,1.6) {\textcolor{blue}{$\Advg$}};
					\draw[->,color = red,dashed] (-2.25,2.1) -- (0,2.1) node[above] {\textcolor{purple}{$\id$}} -- (2.25,2.1);
					\draw[->,color = blue] (2.25,1.9) -- (-2.25,1.9);
					\node at (3,2) {$\germs$};
					\draw[->,color = red,dashed] (-2.5,1.8) -- (-.3,1.1);
					\node at (-.65,1.6) {\textcolor{red}{$\Advg$}};
					\draw[->, color=blue] (2.75,1.75) -- (2.75,0.5) node[left] {\textcolor{blue}{$\Sum$}} -- (2.75,-.75);
					\node at (0,1) {$\NN$};
					\draw[->,color = blue] (0,.8) -- (0,.2);

					\node at (6,1.25) {\textcolor{red}{\dashuline{dash}: $\equivW \CCN$}};
					\node at (6,-.25) {\textcolor{blue}{\underline{line}: $\equivW \id_{\Baire}$}};
				\end{tikzpicture}
				\caption{The results of \Cref{thm:analytic rep,thm:main germs,thm:main functions} and \Cref{lemma:computing germ}.}\label{figure:reductions}
			\end{figure}
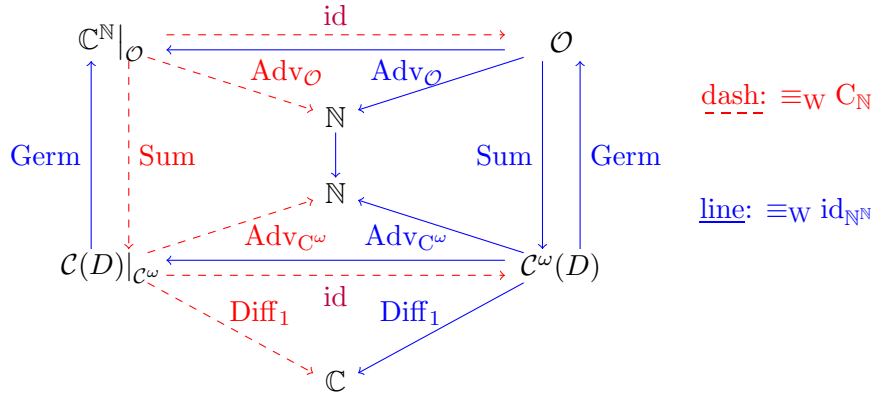

	\section{Polynomials}\label{sec:polynomials}
		\subsection{Polynomials as finite sequences}\label{sec:sub:polynomials as finite sequences}

			Consider the set $\CC[X]$ of polynomials with complex coefficients in one variable $X$.
			There are several straightforward ways to represent polynomials.
			The first one that comes to mind is to represent a polynomial by a finite list of complex numbers.
			One can either demand the length of the list to equal the degree of the polynomial or just to be big enough to contain all of the non-zero coefficients.
			Since the first option fails to make basic operations like addition of polynomials computable, we choose the second option.

			\begin{definition}
				Let $\CC[X]$ denote the \demph{represented space of polynomials}, where the representation is defined as follows: $p\in\Baire$ is a ${\CC[X]}$-name of $P$ if $p(0)\geq \deg(P)$ and $n\mapsto p(n+1)$ is a $\CC^{p(0)+1}$-name of the first $p(0)+1$ coefficients of $P$.
			\end{definition}

			Let $\CC_m[X]$ denote the set of monic polynomials over $\CC$, i.e. the polynomials with leading coefficient equal to one.
			Make $\CC_m[X]$ a represented space by restricting the representation of $\CC[X]$.
			Monic polynomials are important because it is possible to compute their roots -- albeit in an unordered way.
			To formalize this define a representation of the disjoint union $\mathbb{C}^\times := \coprod_{n \in \mathbb{N}} \CC^n$ as follows:
			A function $p$ is a name of $x\in \mathbb{C}^\times$ if and only if $x\in \CC^{p(0)}$ and $n\mapsto p(n+1)$ is a $\CC^{p(0)}$ name of $x$.
			Note that the construction of the representation of $\CC[X]$ is very similar.
			The only difference being that in $\mathbb{C}^\times$ vectors with leading zeros are not identified with shorter vectors.

			Now, the task of finding the zeros in an unordered way can be formalized by computing the multivalued function that maps a polynomial to the set of lists of its zeros, each appearing according to its multiplicities:
			\[ \Zeros: \subseteq \CC[X] \mto \CC^\times, \quad P \mapsto \left\{(a_1,\ldots,a_{\deg(P)})\mid \exists \lambda:P=\lambda\prod_{k=1}^{\deg(P)}(X-a_k)\right\} \]
			The importance of $\CC_m[X]$ is reflected in the following well known lemma:

			\begin{lemma}\label{resu:computable operations on the polynomials}
				Restricted to $\CC_m[X]$ the mapping $\Zeros$ is computable.
			\end{lemma}

			\begin{proof}[Proof sketch]
				This is well known.
				A nice description of an algorithm to do this can for instance be found in \cite{MR1905263}, although algorithms were known a lot longer.
				We only sketch how to find out the degree, which is the number of zeros of the polynomial and therefore the first step towards computing the set of zeros as element of $\CC^\times$.
				Get an approximation to each of the coefficients with precision $\frac 12$.
				Since the highest coefficient will be one, it can be found from this approximation.
			\end{proof}

			The main difficulty in computing the zeros of an arbitrary polynomial is to find its degree.
			A polynomial of known degree can be converted to a monic polynomial with the same zeros by scaling.
			On $\CC[X]$ consider the following functions:
			\begin{itemize}
				\item $\deg$: The function assigning to a polynomial its degree.
				\item $\dbnd$: The multivalued function where an integer is a valid return value if and only if it is an upper bound of the degree of the polynomial.
			\end{itemize}
			$\dbnd$ is computable by definition of the representation of $\CC[X]$.
			$\deg$ is not computable on the polynomials, however, from the proof of Lemma \ref{resu:computable operations on the polynomials} it follows:
			\begin{lemma}\label{resu:deg on the monics}
				The degree mapping is computable when restricted to the monic polynomials.
			\end{lemma}
			The next result classifies finding the degree, turning a polynomial into a monic polynomial and finding the zeros to be Weihrauch equivalent to $\lpo^*$.

			\begin{proposition}
				The following are Weihrauch-equivalent to $\lpo^*$:
				\begin{itemize}
					\item $\deg$, that is the mapping from $\CC[X]$ to $\NN$ defined in the obvious way.
					\item $\Monic$, that is the mapping from $\CC[X]$ to $\CC_m[X]$ defined on the non-zero polynomials by
					\[ P = \sum_{k=0}^{\deg(P)} a_k X^k \mapsto \sum_{k=0}^{\deg(P)}\frac{a_k}{a_{\deg(P)}}X^k. \]
					\item $\Zeros : \subseteq \CC[X] \to \mathbb{C}^\times$, mapping a non-zero polynomial to the set of its zeros, each appearing according to its multiplicity.
				\end{itemize}
			\end{proposition}

			\begin{proof}
				Note that  $\lpo^*$ is Weihrauch equivalent to the function $\min:\Baire \to \NN$ by \Cref{resu:lpo*}.
				Proceed by building a chain of Weihrauch equivalences:
				\begin{description}
					\item[$\min\equivW\deg$:]
					To show\footnote{This direction of the proof was simplified based on a suggestion by an anonymous referee.} $\min\leqW\deg$, note that given $p \in \Baire$ and $n \leq p(0)$, we can compute $a_n$ defined by $a_n = 2^{-\min \{i \mid p(0) - p(i) = n\}}$, where we understand $a_n = 0$ if $\{i \mid p(0) - p(i) = n\}  =\emptyset$. Subsequently we can compute the polynomial $P := \sum_{i=0}^{p(0)} a_ix^i$, and find that $\min p = p(0) - \deg P$.

          On the other hand to see $\deg\leqW\min$ let $p$ be a $\CC[X]$-name of a polynomial $P$.
          Set $H(p)(0):= p(0)$ and let $H(p)(n)$ be the minimal number such that the $2^{-n+1}$-approximation of the polynomial is consistent with $\deg(P)=p(0)-m$.
          Apply $\min$ to this function to get $p(0)-\deg(P)$.
          Thus, the post-processor $K$ can be chosen as $K(p,q) = p(0)-q(0)$
					\item[$\deg\equivW\Monic$:]
					For $\deg\leqW\Monic$ let the pre-processor be the identity.
					Applying $\Monic$ to the input will result in a monic polynomial of the same degree.
					Let the post-processor be the second projection composed with a realizer of the degree mapping on the monic polynomials that can be chosen computable by \Cref{resu:deg on the monics}.

					$\Monic\leqW\deg$ is obvious since division by a non-zero number is computable.
					\item[$\Monic\equivW\Zeros$:]
					To see that $\Monic\leqW\Zeros$ note, that from approximations to a vector $(y_0, y_1, \ldots, y_n)$ of all the zeros of a polynomial $P$ approximations to the coefficients of $\Monic(P)$ can be computed via
					\[ \Monic(P) = \prod_{k \leq n} (X - y_k). \]

					For the opposite direction note, that $P$ and $\Monic(P)$ have the same set of zeros and that the set of zeros can be computed from $\Monic(P)$ by \Cref{resu:computable operations on the polynomials}.
				\end{description}
			\end{proof}

		\subsection{Polynomials as functions}\label{sec:sub:polynomials as functions}

			As polynomials induce analytic functions on the unit disk, the representations of $\analytic$ and $\cont$ can be restricted to the polynomials.
			The represented spaces that result from this are $\analytic|_{\CC[X]}$, resp.\ $\cont|_{\CC[X]}$.
			Here, the choice of the unit disk $D$ as domain seems arbitrary: A polynomial defines a continuous resp.\ analytic function on the whole space.
			The following proposition can easily be checked to hold whenever the domain contains an open neighborhood of zero and, since translations are computable with respect to all the representations we consider, if it contains any open set.

			Denote the versions of the degree resp.\ degree bound functions that take continuous resp.\ analytic functions by $\deg_{\cont}$, $\dbnd_{\cont}$ resp. $\deg_{\analytic}$, $\dbnd_{\analytic}$.
			When polynomials are regarded as functions, resp.\ analytic functions, these maps become harder to compute.

			\begin{theorem}\label{resu:polynomials as analytic functions}
				The following are Weihrauch-equivalent:
				\begin{itemize}
					\item $\CCN$, that is: Closed choice on the naturals.
					\item $\dbnd_{\analytic}$, that is: Given an analytic function which is a polynomial, find an upper bound of its degree.
					\item $\deg_{\analytic}$, that is: Given an analytic function which is a polynomial, find its degree.
				\end{itemize}
			\end{theorem}
			\begin{proof}
				Build a circle of Weihrauch reductions:
				\begin{description}
					\item[$\CCN\leqW\dbnd_{\analytic}$:] Use \Cref{lemma:cn} and reduce to $\Bound$ instead.
					Thus, let $p$ be an enumeration of some bounded subset of the natural numbers.
					Define a polynomial $P$ as follows:
					\[ P(X) := \sum_{n\in\NN} 2^{-(n + p(n))} X^{p(n)}. \]
					One readily verifies that a $\analytic$-name of the function $f$ corresponding to $P$ can be computed from $p$: A $\cont$-name of $f$ is easy to get hold of as the coefficients fall fast enough with $n$, and it is easy to check that $2$ is an allowed value of $\Advc(f)$.
					Let the pre-processor $H$ be a realizer of this assignment.

					Obviously $\dbnd_{\analytic}(f)$ is an upper bound of the set enumerated by $p$.
					This means that the choice $K(p,q):=q$ for the post-processor results in a Weihrauch reduction.
					\item[$\dbnd_{\analytic}\leqW\deg_{\analytic}$:]
					Is trivial: Using the identity as pre-processor and the second projection as post-processor will do.
					\item[$\deg_{\analytic}\leqW\CCN$:]
					By \Cref{lemma:cn} replace $\CCN$ with $\max$.
					Let $p$ be a $\analytic$-name of the function corresponding to some polynomial $P$.
					Shifting the name will result in a $\cont$-name of $P$ and by \Cref{lemma:computing germ} a $\CC^\NN$-name $q$ of the series of coefficients of $P$ can be computed from this.
					Let $d_n$ denote the enumeration of the rational elements of $\CC$ that was fixed for the definition of the representation of $\CC$.
					Define the pre-processor $H$ as follows:
					\[ H(p)(\langle m,n\rangle) := \begin{cases} m+1 &\text{ if }\abs{d_{q(\langle m,n\rangle)}} > 2^{-n} \\ 0 & \text{ otherwise}.\end{cases} \]
					This pre-processor is computable and $H(p)$ enumerates the set of indices $k$ such that $a_k$ is not zero.
					Therefore, applying $\max$ will result in the degree of the polynomial and $K(p,q):=q$ can be chosen as post-processor of a Weihrauch reduction.
				\end{description}
			\end{proof}

			From the proof of the previous theorem it can be seen, that stepping down from analytic to continuous functions is not an issue.
			For sake of completeness we add a slight tightening of the third item of \Cref{thm:main functions} and state this as theorem:

			\begin{theorem}\label{resu:polynomials as continuous functions}
				The following are Weihrauch-equivalent to $\CCN$:
				\begin{itemize}
					\item $\deg_{\cont}$, that is: Given a continuous function which happens to be a polynomial, find its degree.
					\item $\dbnd_{\cont}$, that is: Given an analytic function which happens to be a polynomial, find an upper bound of its degree.
					\item $\Advc|_{\CC[X]}$, that is: Given a continuous function which happens to be a polynomial, find the constant needed to represent it as analytic function.
				\end{itemize}
			\end{theorem}
			\begin{proof}
				Weihrauch equivalence of the first two bullets to $\CCN$ follows directly from the proofs of \Cref{resu:polynomials as analytic functions}.
				For the last item first note, that the Weihrauch reduction $\Advc\leqW\CCN$ constructed in \Cref{thm:main functions} is also a Weihrauch reduction showing $\Advc|_{\CC[X]}\leqW\CCN$.
				This is generally true for restrictions.
				On the other hand, the sequence $f_n$ of analytic functions in the proof of the reduction $\CCN\leqW\Advc$ in the same theorem may be replaced by rational polynomials that approximate the functions and their derivative well enough.
				This way, the constructed function $f$ is a polynomial and the reduction a Weihrauch reduction to the restriction $\Advc|_{\CC[X]}$.
			\end{proof}

			$\dbnd_{\analytic}$ may be regarded as the advice function of $\CC[X]$ over $\analytic$: The representation where a function $p$ is a name of a polynomial $P$ if and only if $p(0) = \dbnd_{\analytic}$ and $n\mapsto p(n+1)$ is a $\analytic$-name of $P$ is computationally equivalent to the representation of $\CC[X]$.
			The same way, $\dbnd_{\cont}$ can be considered an advice function of $\CC[X]$ over $\cont$.

			\Cref{figure:reductions for polynomials} illustrates \Cref{resu:computable operations on the polynomials}, \Cref{resu:lpo*} and \Cref{resu:polynomials as continuous functions,resu:polynomials as analytic functions}.
			\begin{figure}
				\centering
				\begin{tikzpicture}
					\node at (-5.75,0) {$\deg$};
					\node at (-5,0) {$\NN$};
					\node at (0,1.5) {$\CC_m[X]$};
					\draw[color=blue,->] (-.75,1.5) -- (-4.75,.15);
					\draw[color=blue,->] (.75,1.5) -- (4.75,.15);
					\draw[color=blue,->] (.2,1.3) -- (.2,.7);
					\node at (.45, 1) {\textcolor{blue}{$\id$}};
					\node at (0,.5) {$\CC[X]$};
					\draw[color=black,dotted,->] (-.75,.5) -- (-4.75,.05);
					\draw[color=blue,->] (.75,.5) -- (4.75,.05);
					\draw[color=blue,->] (.2,.3) -- (.2,-.3);
					\node at (.45, 0) {\textcolor{blue}{$\id$}};
					\draw[color=black,dotted,->] (-.2,.7) -- (-.2,1.3);
					\node at (-.75, 1) {\textcolor{black}{$\Monic$}};
					\node at (0,-.5) {$\analytic|_{\CC[X]}$};
					\draw[color=red,dashed,->] (-1,-.5) -- (-4.75,-.05);
					\draw[color=red,dashed,->] (1,-.5) -- (4.75,-.05);
					\draw[color=blue,->] (.2,-.7) -- (.2,-1.3);
					\node at (.45, -1) {\textcolor{blue}{$\id$}};
					\draw[color=red,dashed,->] (-.2,-.3) -- (-.2,.3);
					\node at (-.45, 0) {\textcolor{red}{$\id$}};
					\node at (0,-1.5) {$\cont|_{\CC[X]}$};
					\draw[color=red,dashed,->] (-1,-1.5) -- (-4.75,-.15);
					\draw[color=red,dashed,->] (1,-1.5) -- (4.75,-.15);
					\draw[color=red,dashed,->] (-.2,-1.3) -- (-.2,-.7);
					\node at (-.45, -1) {\textcolor{red}{$\id$}};
					\node at (5,0) {$\NN$};
					\node at (5.75,0) {$\dbnd$};
					\node at (8,1) {\textcolor{red}{\dashuline{dash}: $\equivW \CCN$}};
					\node at (8,0) {\textcolor{black}{\dotuline{dots}: $\equivW \lpo^*$}};
					\node at (8,-1) {\textcolor{blue}{\underline{line}: $\equivW \id_{\Baire}$}};
				\end{tikzpicture}
				\caption{The result of \Cref{resu:computable operations on the polynomials}, \Cref{resu:lpo*} and \Cref{resu:polynomials as continuous functions,resu:polynomials as analytic functions}.}\label{figure:reductions for polynomials}
			\end{figure}
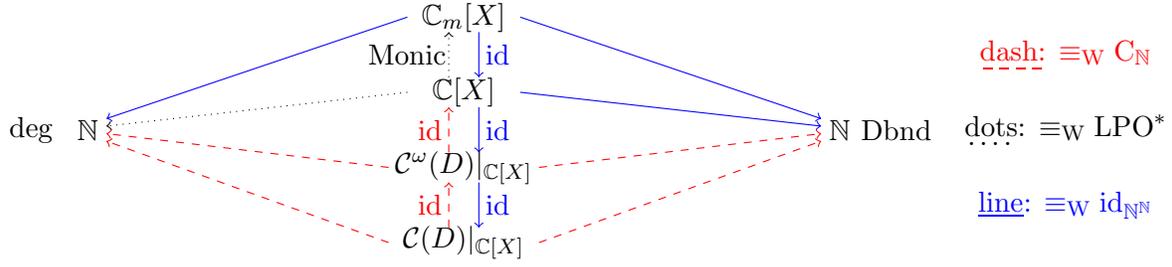

	\section{Test function spaces}\label{sec:schwartz}
		This section considers three spaces of test functions as a final example.
		\subsection{The spaces $\EE$, $\SF$ and $\TF$}
			Consider the spaces
			\[ \EE:=C^\infty(\mathbb{R}) \]
			of smooth functions,
			\[ \SF:=\big\{f\in \EE\mid \forall n,m \ \ \exists C \ \ \forall x: \abs{x^nf^{(m)}(x)}\leq C\big\} \]
			of Schwartz functions and
			\[ \TF:= C^\infty_0(\mathbb{R}) \]
			of bump functions, i.e.~those smooth functions that are zero outside some compact set.
			We use the slightly less common name \lq bump functions\rq\ for $\TF$ instead of the standard name \lq test functions\rq\ as all three spaces are called in \lq spaces of test functions\rq\ and $\TF\subseteq \SF\subseteq \EE$.
			The standard example of a function from $\TF$ is listed in \Cref{ex:a bump function} below.

			These spaces are in particular relevant as their dual spaces with respect to the topologies introduced below are the space $\TF'$ of distributions, $\SF'$ of tempered distributions and $\EE'$ of distributions with compact support.
			The spaces $\TF,\SF$ and $\EE$ are complete locally convex spaces.
			Recall that a topological vector space is called \demph{locally convex} if its topology is the initial topology of a family $(\norm{\cdot}_i)_{i\in I}$ of semi-norms.
			Set  $I:=\NN\times \NN$.
			In the case of $\EE$ the semi-norms
			\[ \norm f^\EE_{N,m} := \sup_{|x| \leq N}\abs{f^{(m)}(x)} \]
			can be used.
			For $\SF$ use the semi-norms defined via
			\[ \norm f ^\SF_{d,m} := \sup_{x \in \mathbb{R}} |x^df^{(m)}(x)|. \]
			Finally note that $\TF$ can be regarded as the union of all the spaces $\TF_K$ of smooth functions with support contained in the compact set $K$.
			The Spaces $\TF_K$ can be regarded topological vector spaces as the space $\EE$ above.
			Define a collection of semi-norms on $\TF$ as follows: A semi-norm $\norm\cdot$ on $\TF$ is contained in the collection if and only if it restricts to a continuous semi-norm on each of the spaces $\TF_K$.

			With respect to the locally convex topologies defined above the inclusions $\iota_\TF^\SF:\TF\hookrightarrow\SF$ and $\iota_\SF^\EE:\SF\hookrightarrow \EE$ are continuous.
			The corresponding subspace topologies, however, are strictly coarser.
			The index set of the families of semi-norms on $\EE$ and $\SF$ are both $\NN\times\NN$ and can be identified with $\NN$ using the pairing function from the introduction.
			This makes these spaces Fr\'echet spaces.
			Note that a Fr\'echet space can always be equipped with a translation invariant metric by setting
			\[ d(x,y) = \sum_{i \in \mathbb{N}} \frac{2^{-i} \norm{x - y}_i}{\norm{x - y}_i + 1}, \]
			where $(\norm\cdot_i)_{i\in \NN}$ is a countable family of semi-norms inducing the topology.
			On $\TF$ there does not exist a countable family of semi norms that induces the right topology:
			It is not metrizable.

		\subsection{Representing test functions}
			For representing these spaces first turn to $\EE$ and $\SF$, which can be handled as metric spaces.
			For the space $\EE$ of smooth functions choose as dense sub sequence the polynomials with rational coefficients.
			Equip $\EE$ with the corresponding metric representation.
			\begin{lemma}\label{resu:evaluation of the derivatives}
				An element of $\EE$ is computable if and only if the mapping
				\[ \NN\times \RR \to \RR, \quad(m,x) \mapsto f^{(m)}(x) \]
				is computable.
			\end{lemma}
			\begin{proof}
				First assume that the mapping is computable.
				To approximate $f$ by polynomials in the translation invariant metric up to precision $2^{-n}$ first note that $N,m\leq\langle N,m\rangle$.
				Therefore it suffices to approximate all up to the $m$-th derivatives of $f$ on the interval $[-N,N]$ in supremum norm.
				Use the computable Weierstra\ss\ Approximation Theorem to find polynomial approximations of $f^{(m)}$ on $[-N,N]$ to precision $(2N)^{-m}2^{-n}$.
				By the Intermediate Value Theorem $f^{(m-1)}(x) -f^{(m-1)}(y) \leq \norm {f^{(m)}}_\infty \abs{x-y}\leq (2N)^{-m+1}2^{-n}$, and analogously none of the derivatives up to $m$ can vary by more than $2^{-n}$.

				For the other direction just get an upper bound $N$ for $x$, read polynomial approximations to the derivatives from the name of the function $f$ and evaluate them.
			\end{proof}
			The above can seen to be uniform in the sense that it proves the metric representation is computably equivalent to the following one:
			\begin{definition}\label{def:smooth functions reprsentation}
				Let $\EE$ denote the \demph{represented space of smooth functions}, where the representation is defined as follows: A function $q\in\Baire$ is a name of a smooth function $f$ if for all $N,m\in\NN$ it holds that $n\mapsto \varphi(\langle N,m ,n\rangle)$ is a $\C([-N,N])$-name of $f^{(m)}$.
			\end{definition}

			Computability on the space $\SF$ of Schwartz functions is investigated in \cite{zhong}\footnote{Prior to \cite{zhong}, in \cite{washihara} computability on $\SF$ was studied in the style of Pour-El and Richards \cite{MR1005942}.}.
			Note that the rational polynomials are not contained in the space $\SF$, thus they have to be replaced by truncating rational polynomials to rational intervals in a smooth way.
			Writing the corresponding sequence down explicitly is cumbersome, however, it can be done.
			An alternative approach to obtain a representation of $\SF$ is to effectivize the definition directly.
			\begin{definition}\label{def:schwartz representation}
				Let $\SF$ denote the \demph{represented space of Schwartz functions}, where the representation is defined as follows:
				A function $q\in\Baire$ is a name of a Schwartz function $f$ if for all $d,m,k \leq n$ it holds that
				\[ \forall x\in\RR: \abs x\geq q(2n) \quad \Rightarrow\quad\abs{x^df^{(m)}(x)}\leq2^{-k} \]
				and $n\mapsto q(2n+1)$ is a name of $f\in\EE$.
			\end{definition}
			A proof that this representation is computably equivalent to the metric representation of $\SF$ can be found as Lemma 5.3 (2) in \cite{zhong}.
			This representation adds information to an $\EE$-name of a function.
			However, in contrast to the other cases we encountered so far the information added is an element of Baire space and not discrete.
			
			Finally turn to the bump functions:
			The space $\TF$ not being metrizable does not prohibit the existence of a well behaved representation.
			For instance also the space $\analytic$ of analytic functions is also not metrizable.
			The topology of $\TF$ is, however, also not sequential.
			Thus, a representation can not be expected to induce the topology itself but at most its sequentialization.
			The question of how to represent $\TF$ is for instance studied in \cite{zhong,MR2207129}.
			\begin{definition}\label{def:bump function representation}
				Let $\TF$ denote the \demph{represented space of bump functions}, where the representation is defined as follows: A function $q\in\Baire$ is a name of a bump function $f$ if and only if the support of $f$ is contained in $[-q(0),q(0)]$ and $n\mapsto q(n+1)$ is a name of $f\in\EE$.
			\end{definition}
			\noindent\begin{minipage}{.55\textwidth}
				Again, the representation of $\TF$ arises from subspace representation of $\EE$ by enriching with discrete advice.
				The advice function given is by
				\[ \adv^\EE_{\TF}(f) = \{k\mid \supp(f)\subseteq [-k,k]\}. \]
				\begin{example}\label{ex:a bump function}
					Let $f$ denote the function defined by
					\[ f(x) := \begin{cases} e^{\frac{x^2}{x^2-1}} &\text{if }\abs x < 1 \\ 0& \text{otherwise} \end{cases} \]
					(compare \Cref{fig:the bump function f})
				\end{example}
			\end{minipage}
			\begin{minipage}{.4\textwidth}
				\begin{tikzpicture}
					\draw[->] (-3.5,0) -- (3.5,0);
					\draw[->] (0,-.2) -- (0,3.5) node[above] {};
					\node at (0,-.4) {$0$};
					\draw (.1,2.5) -- (-.1,2.5);
					\node at (.35,2.75) {$1$};
					\draw (2.5,.1) -- (2.5,-.1) node[below] {$1$};
					\draw (-2.5,.1) -- (-2.5,-.1) node[below] {$-1$};
					\draw[color=blue,thick] (-3,0)--(-2.5,0);
					\draw[color=blue,thick] (3,0)--(2.5,0);				
					\draw[color=blue,thick] plot[samples=100,domain=-0.999:.999] ({2.5*\x},{2.5*exp(\x*\x/(\x*\x-1)});
				\end{tikzpicture}
				
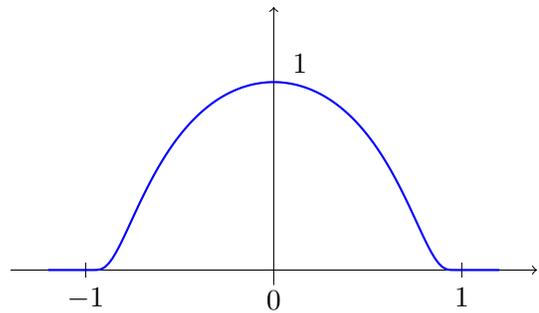
\captionof{figure}{The bump function $f$}\label{fig:the bump function f}
			\end{minipage}

			We need the following:
			\begin{lemma}\label{resu:f and its shifts}
				The function $f$ from \Cref{ex:a bump function} is computable as element of $\TF$.
				For its shifts $f_\lambda(x) := f(x-\lambda)$ it is true that
				\[ \norm{f_\lambda}^{\SF}_{d,m} \leq (\abs\lambda+1)^d (17m)^{4m} .\]
			\end{lemma}
			This can be pieced together from \cite{zhong} or \cite{MR3377508}.
			We give a direct proof.
			\begin{proof}
				That $f$ is computable is clear from its definition.
				A proof that $f$ is infinitely often differentiable is an standard exercise and left to the reader.
				A direct computation of $f'$ and an easy induction show that for $\abs{x}<1$
				\[ f^{(n)}(x) = p_n f(x) (1-x^2)^{-2n}, \]
				where $p_n$ are integer polynomials recursively defined as follows:
				\[ p_0 = 1 \quad\text{and}\quad p_{n+1}:= (1-x^2)^2 p_n' + 2((2n-1)x - 2nx^3)p_n. \]
				These formulas show that the mapping $(m,x)\mapsto f^{(m)}(x)$ is computable.
				The computability of $f$ in $\EE$ now follows from \Cref{resu:evaluation of the derivatives} and to find a computable $\TF$-name it suffices to add $1$ to the front of the computable $\EE$-name.

				Further computations show that $\deg(p_n) = 3n$, that for $n>0$ the absolute value of each of its coefficients is bonded by $(17n-7)^n$.
				From applying l'Hospital's rule it can be seen that $(1-x^2)^{-2n}f(x)\to 0$ for $\abs x\to 1$.
				Thus the suprema can be found by searching for zeros of the derivative.
				They are
				\[ x_1 := \frac1{2n} + \sqrt{\frac1{4n^2} +1} \quad\text{and}\quad x_2 =\frac1{2n} - \sqrt{\frac1{4n^2} +1}\quad\text{with}\quad (1-x_{1/2}^2)^{-2n} f(x_{1/2}) \leq n^{4n}+1. \]
				It follows that
				\[ \bnorm{f^{(n)}}_\infty \leq 3n(17n-7)^n(n^{4n}+1)\leq (17n)^{4n} \]
				and thus
				\[ \norm{f_\lambda}^{\SF}_{d,m} = \sup_{x\in\RR} \abs{x^d f^{(m)}(x-\lambda)} = \sup_{x\in[-1,1]} \abs{(x+\lambda)^df^{(m)}(x)} \leq (\abs\lambda+1)^d(17m)^{4m} \]
			\end{proof}

		\subsection{Inclusions}
			We are now ready to investigate the inclusion maps and their inverses between these spaces.
			Follow \cite[Proposition 5.4]{zhong} to see that the inclusion $\iota_{\TF}^\SF : \TF\hookrightarrow\SF$ is computable:
			Given a $\TF$-name $q$ of $F$ obtain the values of an $\SF$-name of $f$ (compare \Cref{def:schwartz representation}) on even numbers by $2n\mapsto q(0)$ and the values on odd numbers by $2n+1\mapsto q(n+1)$.
			That the map $\iota_{\SF}^\EE :\SF \hookrightarrow \EE$ is computable follows from comparing the metric representations.
			Note that $\iota_\TF^\SF$ and $\iota_\SF^\EE$ are injective thus they have unique partial sections $\pi_\SF^\TF$ and $\pi_\EE^\SF$ with domains $\TF\subseteq \SF$ resp.\ $\SF\subseteq \EE$.
			As mentioned before the subspace topologies are strictly finer, thus the sections are not continuous.

			\begin{proposition}
				The following are Weihrauch-equivalent to $\C_\mathbb{N}$:
				\begin{itemize}
					\item $\pi_\SF^\TF: \subseteq \SF\to \TF$, the partial inverse of the embedding $\TF\hookrightarrow \SF$.
					\item $\pi_\EE^\TF := \pi_\SF^\TF\circ\pi_\EE^\SF:\subseteq \EE \to \TF$, the partial inverse of the embedding $\TF\hookrightarrow \EE$.
				\end{itemize}
				\begin{proof}
					Exhibit a cycle of Weihrauch reductions:
					\begin{description}
						\item[$\C_\mathbb{N} \leqW \pi_\SF^\TF$:] Recall, that by Lemma \ref{lemma:cn} the function $\CCN$ may be replaced the function $\Bound:\subseteq \OO(\NN) \mto\NN$ defined on the finite sets in the obvious way.

						Recall the function $f$ from \Cref{ex:a bump function} and its shifts $f_{\lambda}(x):= f(x-\lambda)$.						
						Let the pre-processor $H$ be a realizer of the function mapping a string function $p$ to the function
						\[ g = \sum_{i\in \NN} 2^{-i-1}(2p(i)+1)^{-i}(17 i)^{-4 i} f_{2p(i)}. \]
						Let $g_k$ denote the $k$-th partial sum of $g$.
						Thus
						\begin{align*}
							d(g,g_k) & = \sum_{\langle d,m\rangle\in\NN} 2^{-\langle d,m\rangle} \frac{\norm{g-g_k}^\SF_{d,m}}{\norm{g-g_k}^\SF_{d,m}+1} \\
							& \leq \sum_{\footnotesize\begin{matrix} \langle d,m\rangle\in\NN \\ d\leq k\text{ and }m\leq k \end{matrix}} 2^{-\langle d,m\rangle} \norm{g-g_k}^\SF_{d,m} + \sum_{\footnotesize\begin{matrix} \langle d,m\rangle\in\NN \\ d> k\text{ or }m> k \end{matrix}} 2^{-\langle d,m\rangle}.
						\end{align*}
						Here the latter sum is obviously smaller than $2^{-k-1}$ and the first can be estimated to be smaller by using the definition of $g_k$ and the estimate of the semi-norms of $f_\lambda$ from \Cref{resu:f and its shifts}.
						Now the sequence $(f_i)_{i\in\NN}$ is a computable element of $\TF^\NN$ and therefore also of $\SF^\NN$ and the sequence $g_k$ can be computed from this sequence and $p$.
						This suffices to compute a $\SF$-name of $g$ from $p$.

						From the definition of $g$ it is clear that any bound on the support of $g$ is also a bound for the values of $p$ and therefore a valid return value of $\Bound$.
						The post-processor $K$ can be chosen as the second projection.
					\item[$\pi_\SF^\TF \leqW \pi_\EE^\TF$:] Follows from the computability of $\pi_\EE^\SF$.
					\item[$\pi_\EE^\TF \leqW \CCN$:]
						Assume we are given an $\EE$-name of some $f\in\DD$.
						Let $(q_n)_{n\in\NN}$ be a standard enumeration of the rationals.
						Let the pre-processor $H$ map $f$ to the set the set
						\[ H(p) := \{K\mid\exists n \forall m: f(q_m) > 2^{-n} \Rightarrow q_m\leq K\}. \]
						This set is not empty since any bound on the support of $f$ fulfills the condition for all $n$.
						On the other hand, if the condition is violated by some $m$, then we can get a polynomial $\frac{2^{-n}-f(q_m)}2$-approximation valid on $[-q_m,q_m]$ from the $\EE$-name and evaluate it on $q_m$ with this the same precision to witness the violation.
						Therefore $H$ is computable.

						Applying $\CCN$ to the set returns a bound of the support and thus the post-processor $K$ can be chosen to be the projection to the second argument.
					\end{description}
				\end{proof}
			\end{proposition}

			\begin{theorem}
				For the partial inverse $\pi_{\EE}^\SF:\subseteq\EE\to \SF$ of the inclusion $\SF\hookrightarrow \EE$ it holds that $\Wd{\pi_{\EE}^\SF}=\lim$, i.e. that $\pi_\EE^\SF \equivW \CCN^\NN$.
				\begin{proof}
					First prove $\pi_\EE^\SF\leqW\CCN^\NN$.
					To obtain a $\SF$-name of $f$ from a $\EE$-name it necessary to find a sequence $a_n$ such that for all $d,m,k\leq n$ it holds that $\abs{x^df^{(m)}(x)}\leq 2^{-k}$ for all $x$ with $\abs{x}\geq 2^{a_n}$.
					Each $a_n$ can be found using one instance of $\CCN$:
					Note that the inequalities above are fulfilled for all real numbers if and only if they are fulfilled for all rational numbers.
					Therefore, it is possible to enumerate those natural numbers $b$ for which the condition is not fulfilled by searching for a rational counterexample.
					It follows that computing the sequence $a_n$, and therefore also an $\SF$-name of $f$, from an $\EE$-name of $f$ is Weihrauch reducible to $\CCN^\NN$.

					For the other direction by \Cref{lemma:count} it suffices to prove $\Bound^\NN\leqW\pi_\EE^\SF$.
					Note that $\Bound^\NN$ produces from $p\in\Baire$ such that
					\[ \forall m \in \NN :\max\{p(\langle n,m\rangle)\mid n\in\NN\} \]
					a $q\in\Baire$ such that $q(m)$ is a bound of the maximum above.

					Let $f$ be the function from \Cref{ex:a bump function} and $f_i$ its integer shifts.
					For $i,k\in\NN$ let $m_{i,k}$ denote the smallest integer such that $p(\langle m_{i,k},k\rangle)=i$.
					Consider the following sequence of functions:
					\[ g_k := \sum_{i\in\mathrm{img}(p)}\max\{x,2\}^{-k} f_{m_{i,k}+i} \]
					A $\EE^\NN$-name of this function sequence is computable from $p$ since $(f_i)$ is computable in $\DD^\NN$ and each restriction of $g_k$ to $[-N,N]$ only depends on the first $N$ values of $p$.
					Let the preprocessor $H$ be a realizer of this mapping.

					Due to the assumption $p\in\dom(\Bound^\NN)$ the functions $g_k$ have compact support.
					Furthermore, from the bounds from \Cref{resu:f and its shifts} it is easy to see that the function
					\[ g:= \sum_{k\in\NN} g_k \quad\text{fulfills}\quad \norm{g}^\SF_{d,m}\leq \sum_{i\in\mathrm{img}(p)} (1+m_{i,k} + i)^d(17m)^{4m} < \infty \]
					and is therefore contained in $\SF$.

					On the other hand all the values of the $g_k$ are positive and therefore
					\[ g(m_{i,k}+p(\langle i,k\rangle))\geq g_k(m_{i,k}+p(\langle i,k\rangle)) \geq 2^{-k} \]
					Therefore, if $q$ is an $\SF$-name of $g$, then for all $i$ it holds that $q(2k) \geq p(\langle i,k\rangle)$ and $q(2\cdot)$ is a valid return value for $\Bound^\NN$.
				\end{proof}
			\end{theorem}
			
			Another important representative of the degree $\lim$ is the Turing jump $J : \Cantor \to \Cantor$ mapping $p$ to the Halting problem relative to $p$.
			As shown in \cite{brattka2}, $J \in \lim$ (essentially by a uniform version of Shoenfields Limit Lemma).
			This has the important consequence that whenever $\Wd{T} = \lim$, then there is some computable point $x \in \dom(T)$ such that $T(x)$ computes the Halting problem \cite{brattka11}.

			\begin{corollary}
				There exists a function $f : \mathbb{R} \to \mathbb{R}$ which is computable as an element of $\EE$, and which also is an element of $\SF$, but as the latter, computes the Halting problem.
			\end{corollary}

	\section{Conclusions \& Outlook}
		We have seen that some care is required when formulating statements about computability with respect to some more complicated function spaces in analysis:
		Computing a continuous function (that happens to be analytic) is easier than computing the same function as an analytic function, etc.
		For most the distinctions we have investigated, these differences are rather small:
		As the degree of $\C_\mathbb{N}$ preserves computability of points, an individual analytic function is a computable analytic function iff it is a computable continuous function.
		For Schwartz functions, however, the situation is different, and crucial distinctions already appear at the level of individual functions.

		The relevance of the choice of function spaces for computable analysis has been very prominent in the discussion of the computability of the wave equation:
		In \cite{pourel2,pourel3}, computable parameters were exhibited that forced the solution to take non-computable values at time $1$.
		This constituted a significant challenge for the philosophical discussion about computability and physics.
		A resolution was then offered in \cite{zhong2} by demonstrating that the solution operator for the wave equation is computable after all -- if one chooses the correct function spaces.

		The examples we have studied in this paper are by far not all that deserve attention.
		Based partially on the results in Section \ref{sec:schwartz}, one could contrast continuous functions, distributions and tempered distributions.
		Similarly, the relationship between continuous functions and $\mathcal{L}^1$ functions that happen to be continuous should be clarified.
		It is known, however, that translating from a continuous real function that has a continuous derivative to a $C^1(\mathbb{R})$ function is equivalent to $\lim$ \cite{stein}.

		We should point out that many of the results proved in \Cref{sec:analytic} work for more general domains:
		\Cref{lemma:computing germ} generalizes to any computable point of the interior of an arbitrary domain.
		It can be made a uniform statement by including the base point of a germ.
		In this case for the proof to go through computability of the distance function of the complement of the domain of the analytic function is needed.

		Another example is the part of \Cref{thm:main functions} that says finding a germ on the boundary is difficult.
		In this case a disc of finite radius touching the boundary in a computable point is needed.
		Alternatively, a simply connected bounded Lipshitz domain with a computable point in the boundary can be used.
		Also in this case it seems reasonable to assume that a uniform statement can be proven.

		\bibliographystyle{alpha}
		\bibliography{../bib}

	\subsubsection*{Acknowledgements.}
		The work has benefited from the Marie Curie International Research Staff Exchange Scheme \emph{Computable Analysis}, PIRSES-GA-2011- 294962. The first author was supported partially by the ERC inVEST (279499) project. The second author was supported by the International Research Training Group 1529 \lq Mathematical Fluid Dynamics\rq\ funded by the DFG and JSPS.

We are grateful to Matthias Schr\"oder for discussions on the subject.

\end{document}